\documentclass[11pt]{article}

\usepackage{amssymb,latexsym}
\usepackage{amsmath}
\usepackage{amsfonts}
\usepackage{amsopn}
\usepackage{amsthm}
\usepackage{epsfig}
\usepackage{lscape}
\usepackage{multirow}
\usepackage{url}
\usepackage[bookmarks=true]{hyperref}
\usepackage{pgf,tikz}
\usepackage{float}
\usepackage{enumerate}
\usetikzlibrary{arrows}
\evensidemargin\oddsidemargin \topmargin-15mm

\newcounter{NN}
\newtheorem{remark}[NN]{Remark}
\newtheorem{proposition}[NN]{Proposition}
\newtheorem{theorem}[NN]{Theorem}
\newtheorem{corollary}[NN]{Corollary}

\def\Z{\mathbb{Z}}

\def\Z{\mathbb{Z}}

\begin{document}
\title{Linear degree  growth in  lattice equations}
\author{ Dinh T. Tran and  John A.~G. Roberts\\
School of Mathematics and Statistics,
University of New South Wales, Sydney 2052 Australia
}
\maketitle
\begin{abstract}
We  conjecture recurrence relations satisfied by the degrees of some linearizable lattice equations. This helps  to prove linear growth of these equations.
We then use these recurrences to search for lattice equations that have linear growth and hence are linearizable.
\end{abstract}

\section{Introduction}
Discrete integrable systems have been a topic of many studies for the last two decades. One characteristic of discrete integrable systems is low complexity.
Complexity of a map or an equation can be measured through the so-called algebraic entropy. It has been used as an integrability detector \cite{Halburd09HowtoDetect, Kamp_growth, Tremblay,  Viallet, Viallet2015}. It is believed that if a discrete map or lattice equation has  vanishing entropy, i.e. degrees of iterations of the map (or equation) in terms of initial variables  grow polynomially, then it is integrable. On the other hand, most discrete integrable maps or lattice equations have vanishing entropy.

 Algebraic entropy of  a map or an equation is often calculated   as follows. One can introduce projective coordinates for each variable and write a map or an equation as a rule in those coordinates. By looking at the actual degrees (degrees after canceling all common factors or gcd) at each iteration, we obtain the degree sequence of the rule. The next step is to find a generating function for this sequence and hence we can calculate algebraic entropy for the rule.
However, it seems that the underlying reason for many integrable maps (equations) to have vanishing entropy has not been focused until recently \cite{Kamp_growth, KankiExtendedHV, RobertsTran, Viallet2015}. In \cite{RobertsTran} we were able to prove  polynomial growth of a large class of lattice equations (autonomous and non-autonomous) subject to a conjecture. The key ingredient for our  approach was to find a recursive formula for the greatest common factor and then derive a linear  recurrence relation for the actual degrees. The results for lattice equations then can be applied to  mappings obtained as reductions of the  lattice equations (with some exceptions). In general, many known integrable lattice equations were shown to share the same \emph{universal} linear recurrence for the actual degrees which leads to  quadratic growth.

We note that there is a sub-class of integrable lattice equations which are linearizable, i.e. equations can be brought to linear equations or systems after some transformations cf. \cite{DiscreteLiouville, LC_linear_2011, LC_linear, Sahadevan, RJGT}).
One quick test for linearization  is using the algebraic entropy test. Linear  growth of degrees of an equation indicates that  this equation can be linearized.
 One then can use the symmetry approach given in \cite{LC_linear} to transform it to a linear equation or a system of linear equations.
 Similarly to \cite{RobertsTran}, the first question that arises here is what are the recurrence relations for the actual degrees of some known linearizable equations. If there is  such a recurrence, is it related to the recurrence that we have found previously? On the other hand, starting from the recurrence relation that shows quadratic growth, can we find some  specialization of it  that  implies linear growth? Can we then search for candidate equations that satisfy these relations? Finally, can we provide transformations to bring these candidates to linear equations?

In this paper, we will try to answer some of these questions. This paper is organised as follows.  In section 2, we briefly set up all the notations that will be used for the paper. In the next section,  we  present some recurrence relations for the actual degrees to have linear growth, i.e. $(i)H, (ii)H$ and $(iii)H$ of Theorem \ref{T:homogenous}.   We  then  show that the recursive formulas for the gcd of the discrete Liouville equation \cite{DiscreteLiouville}, Ramani-Joshi-Grammaticos-Tamizhmani equation \cite{RJGT} and the $(3,1)$ reduction of pKdV \cite{Kamp_growth} imply these recurrence relations. Finally, a search for examples of equations on quad-graphs that satisfy those relations is carried out.  This is done by constructing the recursive relations  for the gcds based on the linear recurrences given in section 3.  The Ramani-Joshi-Grammaticos-Tamizhmani equation is obtained though this search.  A transformation to bring it to a linear equation will be given following the procedure in \cite{LC_linear}.

\section{Setting}
In this section, we give a procedure to compute degrees of lattice equations defined on a square. This procedure was presented in detail in \cite{RobertsTran}.
We consider a multi-affine equation (i.e. linear on each variable) on the quad-graph
\begin{equation}
\label{E:EQ_square}
Q(u,u_1,u_2,u_{12})=0,
\end{equation}
where $u=u_{l,m}$ and subscripts $_1,\ _2$ denote the shifts in the $l$ and $m$ directions, respectively.
One can solve uniquely for each vertex of this equation. Suppose that we solve for the top right vertex $u_{12}$. By introducing $u=x/z$, we obtain the rule
at the top right vertex
\begin{align*}
x_{12}&=f^{(1)}(x,x_1,x_2,z,z_1,z_2),\\
z_{12}&=f^{(2)}(x,x_1,x_2,z,z_1,z_2).
\end{align*}
Initial values are typically  given either on  the boundary of the first quadrant `corner initial values' $I_1=\{(l,m)\in \Z\times \Z\ | \ lm=0,\  l, \ m\geq 0\}$ or on  the $(1,-1)$ staircase $I_2=\{(n,-n)\ | \ n\in \Z\}\cup \{(n+1,-n)\ | \ n\in\Z\}$. They are given in blue paths in the following figure.
\definecolor{qqqqff}{rgb}{0.,0.,1.}
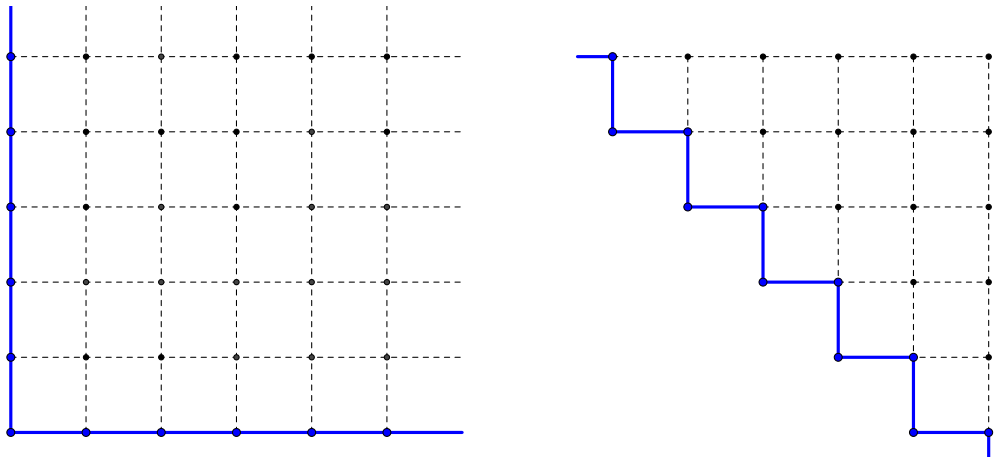
\begin{figure}[H]
\begin{center}
\definecolor{uuuuuu}{rgb}{0.26666666666666666,0.26666666666666666,0.26666666666666666}
\begin{tikzpicture}[line cap=round,line join=round,>=triangle 45,x=1.0cm,y=1.0cm]
\clip(-0.8874187062564846,-0.3177339728839405) rectangle (13.956157422319317,5.668294479833502);
\draw [line width=1.2pt,color=qqqqff] (0.,6.)-- (0.,0.);
\draw [line width=1.2pt,color=qqqqff] (0.,0.)-- (6.,0.);
\draw [line width=0.4pt,dash pattern=on 2pt off 2pt] (0.,3.)-- (6.,3.);
\draw [line width=0.4pt,dash pattern=on 2pt off 2pt] (0.,4.)-- (6.,4.);
\draw [line width=0.4pt,dash pattern=on 2pt off 2pt] (0.,5.)-- (6.,5.);
\draw [line width=0.4pt,dash pattern=on 2pt off 2pt] (1.,0.)-- (1.,6.);
\draw [line width=0.4pt,dash pattern=on 2pt off 2pt] (2.,0.)-- (2.,6.);
\draw [line width=0.4pt,dash pattern=on 2pt off 2pt] (3.,0.)-- (3.,6.);
\draw [line width=0.4pt,dash pattern=on 2pt off 2pt] (4.,0.)-- (4.,6.);
\draw [line width=0.4pt,dash pattern=on 2pt off 2pt] (5.,0.)-- (5.,6.);
\draw [line width=0.4pt,dash pattern=on 2pt off 2pt] (0.,1.)-- (6.,1.);
\draw [dash pattern=on 2pt off 2pt] (0.,2.)-- (6.,2.);
\draw [line width=1.2pt,color=qqqqff] (7.5335130905376655,5.)-- (8.,5.);
\draw [line width=1.2pt,color=qqqqff] (8.,5.)-- (8.,4.);
\draw [line width=1.2pt,color=qqqqff] (8.,4.)-- (9.,4.);
\draw [line width=1.2pt,color=qqqqff] (9.,4.)-- (9.,3.);
\draw [line width=1.2pt,color=qqqqff] (9.,3.)-- (10.,3.);
\draw [line width=1.2pt,color=qqqqff] (10.,3.)-- (10.,2.);
\draw [line width=1.2pt,color=qqqqff] (10.,2.)-- (11.,2.);
\draw [line width=1.2pt,color=qqqqff] (11.,2.)-- (11.,1.);
\draw [line width=1.2pt,color=qqqqff] (11.,1.)-- (12.,1.);
\draw [line width=1.2pt,color=qqqqff] (12.,1.)-- (12.,0.);
\draw [line width=1.2pt,color=qqqqff] (12.,0.)-- (13.,0.);
\draw [dash pattern=on 2pt off 2pt] (8.,5.)-- (9.,5.);
\draw [dash pattern=on 2pt off 2pt] (9.,5.)-- (9.,4.);
\draw [dash pattern=on 2pt off 2pt] (9.,4.)-- (10.,4.);
\draw [dash pattern=on 2pt off 2pt] (10.,4.)-- (10.,3.);
\draw [dash pattern=on 2pt off 2pt] (10.,3.)-- (11.,3.);
\draw [dash pattern=on 2pt off 2pt] (11.,3.)-- (11.,2.);
\draw [dash pattern=on 2pt off 2pt] (11.,2.)-- (12.,2.);
\draw [dash pattern=on 2pt off 2pt] (12.,2.)-- (12.,1.);
\draw [dash pattern=on 2pt off 2pt] (13.,1.)-- (12.,1.);
\draw [dash pattern=on 2pt off 2pt] (13.,1.)-- (13.,2.);
\draw [dash pattern=on 2pt off 2pt] (12.,2.)-- (13.,2.);
\draw [dash pattern=on 2pt off 2pt] (11.,3.)-- (12.,3.);
\draw [dash pattern=on 2pt off 2pt] (12.,3.)-- (12.,2.);
\draw [dash pattern=on 2pt off 2pt] (10.,4.)-- (11.,4.);
\draw [dash pattern=on 2pt off 2pt] (11.,4.)-- (11.,3.);
\draw [dash pattern=on 2pt off 2pt] (9.,5.)-- (10.,5.);
\draw [dash pattern=on 2pt off 2pt] (10.,5.)-- (10.,4.);
\draw [dash pattern=on 2pt off 2pt] (10.,5.)-- (11.,5.);
\draw [dash pattern=on 2pt off 2pt] (11.,5.)-- (11.,4.);
\draw [dash pattern=on 2pt off 2pt] (11.,4.)-- (12.,4.);
\draw [dash pattern=on 2pt off 2pt] (12.,4.)-- (12.,3.);
\draw [dash pattern=on 2pt off 2pt] (12.,3.)-- (13.,3.);
\draw [dash pattern=on 2pt off 2pt] (13.,2.)-- (13.,3.);
\draw [dash pattern=on 2pt off 2pt] (13.,4.)-- (12.,4.);
\draw [dash pattern=on 2pt off 2pt] (13.,4.)-- (13.,3.);
\draw [dash pattern=on 2pt off 2pt] (11.,5.)-- (12.,5.);
\draw [dash pattern=on 2pt off 2pt] (12.,5.)-- (12.,4.);
\draw [dash pattern=on 2pt off 2pt] (12.,5.)-- (13.,5.);
\draw [dash pattern=on 2pt off 2pt] (13.,4.)-- (13.,5.);
\draw [dash pattern=on 2pt off 2pt] (13.,0.)-- (13.,1.);
\draw [line width=1.2pt,color=qqqqff] (13.,0.)-- (13,-0.7325190330435544);
\begin{scriptsize}
\draw [fill=qqqqff] (0.,0.) circle (1.5pt);
\draw [fill=qqqqff] (0.,1.) circle (1.5pt);
\draw [fill=qqqqff] (0.,2.) circle (1.5pt);
\draw [fill=qqqqff] (0.,3.) circle (1.5pt);
\draw [fill=qqqqff] (0.,4.) circle (1.5pt);
\draw [fill=qqqqff] (0.,5.) circle (1.5pt);
\draw [fill=qqqqff] (1.,0.) circle (1.5pt);
\draw [fill=qqqqff] (2.,0.) circle (1.5pt);
\draw [fill=qqqqff] (3.,0.) circle (1.5pt);
\draw [fill=qqqqff] (4.,0.) circle (1.5pt);
\draw [fill=qqqqff] (5.,0.) circle (1.5pt);
\draw [fill=black] (1.,1.) circle (1.0pt);
\draw [fill=black] (1.,3.) circle (1.0pt);
\draw [fill=black] (1.,4.) circle (1.0pt);
\draw [fill=black] (1.,5.) circle (1.0pt);
\draw [fill=uuuuuu] (2.,5.) circle (1.0pt);
\draw [fill=black] (2.,4.) circle (1.0pt);
\draw [fill=uuuuuu] (2.,3.) circle (1.0pt);
\draw [fill=black] (2.,1.) circle (1.0pt);
\draw [fill=black] (3.,3.) circle (1.0pt);
\draw [fill=black] (3.,4.) circle (1.0pt);
\draw [fill=black] (3.,5.) circle (1.0pt);
\draw [fill=black] (4.,5.) circle (1.0pt);
\draw [fill=uuuuuu] (4.,4.) circle (1.0pt);
\draw [fill=uuuuuu] (4.,3.) circle (1.0pt);
\draw [fill=uuuuuu] (5.,3.) circle (1.0pt);
\draw [fill=black] (5.,4.) circle (1.0pt);
\draw [fill=black] (5.,5.) circle (1.0pt);
\draw [fill=uuuuuu] (1.,2.) circle (1.0pt);
\draw [fill=uuuuuu] (2.,2.) circle (1.0pt);
\draw [fill=uuuuuu] (3.,2.) circle (1.0pt);
\draw [fill=uuuuuu] (3.,1.) circle (1.0pt);
\draw [fill=uuuuuu] (4.,2.) circle (1.0pt);
\draw [fill=uuuuuu] (4.,1.) circle (1.0pt);
\draw [fill=uuuuuu] (5.,1.) circle (1.0pt);
\draw [fill=uuuuuu] (5.,2.) circle (1.0pt);
\draw [fill=qqqqff] (8.,5.) circle (1.5pt);
\draw [fill=qqqqff] (8.,4.) circle (1.5pt);
\draw [fill=qqqqff] (9.,4.) circle (1.5pt);
\draw [fill=qqqqff] (9.,3.) circle (1.5pt);
\draw [fill=qqqqff] (10.,3.) circle (1.5pt);
\draw [fill=qqqqff] (10.,2.) circle (1.5pt);
\draw [fill=qqqqff] (11.,2.) circle (1.5pt);
\draw [fill=qqqqff] (11.,1.) circle (1.5pt);
\draw [fill=qqqqff] (12.,1.) circle (1.5pt);
\draw [fill=qqqqff] (12.,0.) circle (1.5pt);
\draw [fill=qqqqff] (13.,0.) circle (1.5pt);
\draw [fill=black] (9.,5.) circle (1.0pt);
\draw [fill=black] (10.,4.) circle (1.0pt);
\draw [fill=black] (11.,3.) circle (1.0pt);
\draw [fill=black] (12.,2.) circle (1.0pt);
\draw [fill=black] (13.,1.) circle (1.0pt);
\draw [fill=black] (13.,2.) circle (1.0pt);
\draw [fill=black] (12.,3.) circle (1.0pt);
\draw [fill=black] (11.,4.) circle (1.0pt);
\draw [fill=black] (10.,5.) circle (1.0pt);
\draw [fill=black] (11.,5.) circle (1.0pt);
\draw [fill=black] (12.,4.) circle (1.0pt);
\draw [fill=black] (13.,3.) circle (1.0pt);
\draw [fill=black] (13.,4.) circle (1.0pt);
\draw [fill=black] (12.,5.) circle (1.0pt);
\draw [fill=black] (13.,5.) circle (1.0pt);
\end{scriptsize}
\end{tikzpicture}
\end{center}
\caption{Initial values $I_1$ (left) and $I_2$ (right) for lattice equations.  \label{F:initial_values}}
\end{figure}
We choose initial values for $x$ and $z$ as polynomials in $w$ of the same degree. Using the rule for the top right vertex $u_{12}$, one can find all the points on the right hand side of these boundaries as polynomials in $w$.
Let ${\rm{gcd}}_{l,m}(w)=\gcd(x_{l,m}(w),z_{l,m}(w))$ and write
\begin{equation}
x_{l,m}(w)={\rm{gcd}}_{l,m}(w)\ \bar{x}_{l,m}(w)\ \mbox{and}\ z_{l,m}(w)={\rm{gcd}}_{l,m}(w)\ \bar{z}_{l,m}(w).
\end{equation}
We denote $d_{l,m}=\deg(x_{l,m})=\deg(z_{l,m})$, $\bar{d}_{l,m}=\deg(\bar{x}_{l,m})=\deg(\bar{z}_{l,m})$ and $g_{l,m}=\deg(\gcd_{l,m})$.
It is easy to see that
\begin{equation}
\label{E:total_degree}
d_{l+1,m+1}=d_{l,m}+d_{l+1,m}+d_{l,m+1}\ \mbox{and}\ d_{l,m}=g_{l,m}+\bar{d}_{l,m}.
\end{equation}
Moreover, we also know that $\gcd_{l,m}(w)\gcd_{l+1,m}(w)\gcd_{l,m+1}(w)|\gcd_{l+1,m+1}(w)$; therefore, we introduce the spontaneous gcd
\begin{equation}
\label{E:bar_gcd_basic}
\overline{\gcd}_{l+1,m+1}=\frac{\gcd_{l+1,m+1}}{\gcd_{l,m}\gcd_{l+1,m}\gcd_{l,m+1}} \ \implies  \ \bar{g}_{l+1,m+1}:=\deg(\overline{\gcd}_{l+1,m+1})=g_{l+1,m+1}-g_{l,m}-g_{l+1,m}-g_{l,m+1}.
\end{equation}
In \cite{RobertsTran} we conjectured a recursive formula for $\gcd_{l,m}$ which led to a  recurrence relation for the actual degrees $\bar{d}_{l,m}$.
For all integrable lattice equations considered in this paper,  the recurrence relation is
\begin{equation}
\label{E:integrable_degree}
\bar{d}_{l-1,m-1}+\bar{d}_{l,m+1}+\bar{d}_{l+1,m}=\bar{d}_{l+1,m+1}+\bar{d}_{l-1,m}+\bar{d}_{l,m-1}.
\end{equation}
\begin{figure}[H]
\begin{center}
\begin{tikzpicture}[line cap=round,line join=round,>=triangle 45,x=1.5cm,y=1.5cm]
\clip(-4.12,2.66) rectangle (0.16,5.54);
\draw [line width=0.4pt,dash pattern=on 2pt off 2pt] (-3.,5.)-- (-3.,3.);
\draw [line width=0.4pt,dash pattern=on 2pt off 2pt] (-3.,3.)-- (-1.,3.);
\draw [line width=0.4pt,dash pattern=on 2pt off 2pt] (-1.,3.)-- (-1.,5.);
\draw [line width=0.4pt,dash pattern=on 2pt off 2pt] (-1.,5.)-- (-3.,5.);
\draw [line width=0.4pt,dash pattern=on 2pt off 2pt] (-2.,5.)-- (-2.,3.);
\draw [line width=0.4pt,dash pattern=on 2pt off 2pt] (-3.,4.)-- (-1.,4.);
\draw[line width=1.2pt,color=red] (-3.,4.)-- (-1.,5.);
\draw [line width=1.2pt,color=red](-1.,5.)-- (-2.,3.);
\draw [line width=1.2pt,color=red](-2.,3.)-- (-3.,4.);
\draw[line width=1.2pt,color=qqqqff] (-2.,5.)-- (-1.,4.);
\draw[line width=1.2pt,color=qqqqff] (-1.,4.)-- (-3.,3.);
\draw[line width=1.2pt,color=qqqqff] (-2.,5.)-- (-3.,3.);
\draw [fill=black] (-3.,5.) circle (1.5pt);
\draw [fill=black] (-3.,3.) circle (1.5pt);
\draw [fill=black] (-1.,3.) circle (1.5pt);
\draw [fill=black] (-1.,5.) circle (1.5pt);
\draw [fill=black] (-2.,5.) circle (1.5pt);
\draw [fill=black] (-2.,3.) circle (1.5pt);
\draw [fill=black] (-3.,4.) circle (1.5pt);
\draw [fill=black] (-1.,4.) circle (1.5pt);
\end{tikzpicture}
\end{center}
\caption{Illustration of the  degree recurrence relation~\eqref{E:integrable_degree}.   \label{F:bar_deg_relation}}
\end{figure}
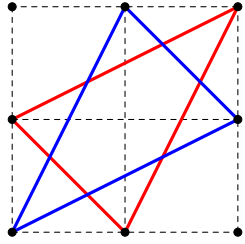
As indicated in Figure \ref{F:bar_deg_relation} the recurrence relation~\eqref{E:integrable_degree} guarantees that the sum of the degrees on the blue vertices and on  the red vertices are the same.

\section{Conditions for linear growth \label{S:conditions}}
In this section, we show that specializations of equation~\eqref{E:integrable_degree} imply linear degree growth.

The equation~\eqref{E:integrable_degree} is equivalent to any  of the rearranged equations below:
\begin{align}
\bar{d}_{l,m}-\bar{d}_{l-1,m}-\bar{d}_{l,m-1}+\bar{d}_{l-1,m-1}&=\bar{d}_{l+1,m+1}-\bar{d}_{l,m+1}-\bar{d}_{l+1,m}+\bar{d}_{l,m}, \label{E:relation1}\\
\bar{d}_{l+1,m}-\bar{d}_{l,m}-\bar{d}_{l,m-1}+\bar{d}_{l-1,m-1}&=\bar{d}_{l+1,m+1}-\bar{d}_{l,m+1}-\bar{d}_{l,m}+\bar{d}_{l-1,m}, \label{E:relation2}\\
\bar{d}_{l,m+1}-\bar{d}_{l,m}-\bar{d}_{l-1,m}+\bar{d}_{l-1,m-1}&=\bar{d}_{l+1,m+1}-\bar{d}_{l+1,m}-\bar{d}_{l,m}+\bar{d}_{l,m-1}\label{E:relation3},
\end{align}
where the RHS in each case is the shift of the LHS in the diagonal, vertical and horizontal directions, respectively.
These equations  are illustrated  in Figure~\ref{F:degeneration} where the $\pm$ denote the coefficients on either sides of each equation and again the weighted sum of the blue vertices equals to the weighted sum of the red vertices (where `weight' refers to the sign of the coefficient).
\begin{figure}[H]
\begin{center}
\definecolor{qqqff}{rgb}{0.,0.,0.}
\begin{tikzpicture}[line cap=round,line join=round,>=triangle 45,x=1.0cm,y=1.0cm]
\clip(-4.860265602465795,3.6333264313320184) rectangle (7.004653850733142,6.95820178137642);
\draw [line width=0.4pt,dash pattern=on 2pt off 2pt] (-4.,6.)-- (-4.,5.);
\draw[line width=1.2pt,color=qqqqff] (-4.,5.)-- (-4.,4.);
\draw [line width=1.2pt,color=qqqqff](-4.,4.)-- (-3.,4.);
\draw [line width=0.4pt,dash pattern=on 2pt off 2pt] (-3.,4.)-- (-2.,4.);
\draw [line width=0.4pt,dash pattern=on 2pt off 2pt] (-2.,4.)-- (-2.,5.);
\draw[line width=1.2pt,color=red] (-2.,5.)-- (-2.,6.);
\draw [line width=0.4pt,dash pattern=on 2pt off 2pt] (-4.,6.)-- (-3.,6.);
\draw[line width=1.2pt,color=red] (-3.,6.)-- (-2.,6.);
\draw[line width=1.2pt,color=qqqqff] (-4.,5.)-- (-3.,5.);
\draw [line width=1.2pt,color=qqqqff](-3.,5.)-- (-3.,4.);
\draw[line width=1.2pt,color=red] (-3.,6.)-- (-3.,5.);
\draw[line width=1.2pt,color=red] (-3.,5.)-- (-2.,5.);
\draw [line width=0.4pt,dash pattern=on 2pt off 2pt] (0.,6.)-- (0.,5.);
\draw [line width=0.4pt,dash pattern=on 2pt off 2pt] (0.,5.)-- (0.,4.);
\draw[line width=1.2pt,color=qqqqff] (0.,4.)-- (1.,4.);
\draw [line width=0.4pt,dash pattern=on 2pt off 2pt] (1.,4.)-- (2.,4.);
\draw [line width=0.4pt,dash pattern=on 2pt off 2pt] (2.,4.)-- (2.,5.);
\draw [line width=0.4pt,dash pattern=on 2pt off 2pt] (2.,5.)-- (2.,6.);
\draw [line width=1.2pt,color=red](2.,6.)-- (1.,6.);
\draw [line width=0.4pt,dash pattern=on 2pt off 2pt] (1.,6.)-- (0.,6.);
\draw [line width=1.2pt,color=red](1.,5.)-- (0.,5.);
\draw[line width=1.2pt,color=qqqqff] (1.,5.)-- (2.,5.);
\draw [line width=0.4pt,dash pattern=on 2pt off 2pt] (1.,5.)-- (1.,4.);
\draw [line width=1.2pt,color=red](0.,5.)-- (1.,6.);
\draw [line width=1.2pt,color=red](2.,6.)-- (1.,5.);
\draw [line width=1.2pt,color=qqqqff](2.,5.)-- (1.,4.);
\draw [line width=1.2pt,color=qqqqff](1.,5.)-- (0.,4.);
\draw [line width=0.4pt,dash pattern=on 2pt off 2pt] (1.,6.)-- (1.,5.);
\draw [line width=1.2pt,color=red](6.,6.)-- (6.,5.);
\draw [line width=1.2pt,color=red](5.,4.)-- (6.,5.);
\draw [line width=1.2pt,color=red](5.,4.)-- (5.,5.);
\draw [line width=1.2pt,color=red](5.,5.)-- (6.,6.);
\draw [line width=1.2pt,color=qqqqff](5.,6.)-- (5.,5.);
\draw [line width=1.2pt,color=qqqqff](5.,6.)-- (4.,5.);
\draw [line width=1.2pt,color=qqqqff](4.,4.)-- (4.,5.);
\draw [line width=1.2pt,color=qqqqff](4.,4.)-- (5.,5.);
\draw [line width=0.4pt,dash pattern=on 2pt off 2pt] (4.,4.)-- (5.,4.);
\draw [line width=0.4pt,dash pattern=on 2pt off 2pt] (5.,4.)-- (6.,4.);
\draw [line width=0.4pt,dash pattern=on 2pt off 2pt] (6.,4.)-- (6.,5.);
\draw [line width=0.4pt,dash pattern=on 2pt off 2pt] (6.,5.)-- (5.,5.);
\draw [line width=0.4pt,dash pattern=on 2pt off 2pt] (5.,5.)-- (4.,5.);
\draw [line width=0.4pt,dash pattern=on 2pt off 2pt] (4.,5.)-- (4.,6.);
\draw [line width=0.4pt,dash pattern=on 2pt off 2pt] (4.,6.)-- (5.,6.);
\draw [line width=0.4pt,dash pattern=on 2pt off 2pt] (5.,6.)-- (6.,6.);
\begin{scriptsize}
\draw [fill=qqqff] (-4.,6.) circle (1.5pt);
\draw[color=blue] (-4.,4.) node[anchor=north] {$+$};
\draw[color=blue] (-3.,4.) node[anchor=north] {$-$};
\draw[color=blue] (-4.,5.) node[anchor=north east] {$-$};
\draw[color=blue] (-3.,5.) node[anchor=north east] {$+$};
\draw[color=red] (-3.,5.) node[anchor=south west] {$+$};
\draw[color=red] (-2.,5.) node[anchor=south west] {$-$};
\draw[color=red] (-2.,6.) node[anchor=south ] {$+$};
\draw[color=red] (-3.,6.) node[anchor=south ] {$-$};
\draw[color=blue] (0.,4.) node[anchor=north] {$+$};
\draw[color=blue] (1.,4.) node[anchor=north] {$-$};
\draw[color=blue] (2.,5.) node[anchor=north west] {$+$};
\draw[color=blue] (1.,5.) node[anchor=north west] {$-$};
\draw[color=red] (1.,5.) node[anchor=south east] {$-$};
\draw[color=red] (0.,5.) node[anchor=south east] {$+$};
\draw[color=red] (2.,6.) node[anchor=south] {$+$};
\draw[color=red] (1.,6.) node[anchor=south] {$-$};
\draw[color=blue] (4.,4.) node[anchor= east] {$+$};
\draw[color=blue] (4.,5.) node[anchor= south east] {$-$};
\draw[color=blue] (5.,5.) node[anchor=south east ] {$-$};
\draw[color=blue] (5.,6.) node[anchor=south east ] {$+$};
\draw[color=red] (5.,4.) node[anchor=north ] {$+$};
\draw[color=red] (5.,5.) node[anchor=north west ] {$-$};
\draw[color=red] (6.,5.) node[anchor=north west ] {$-$};
\draw[color=red] (6.,6.) node[anchor=north west ] {$+$};
\end{scriptsize}
\draw [fill=qqqff] (-4.,5.) circle (1.5pt);
\draw [fill=qqqff] (-4.,4.) circle (1.5pt);
\draw [fill=qqqff] (-3.,4.) circle (1.5pt);
\draw [fill=qqqff] (-2.,4.) circle (1.5pt);
\draw [fill=qqqff] (-2.,5.) circle (1.5pt);
\draw [fill=qqqff] (-2.,6.) circle (1.5pt);
\draw [fill=qqqff] (-3.,6.) circle (1.5pt);
\draw [fill=qqqff] (-3.,5.) circle (1.5pt);
\draw [fill=qqqff] (0.,6.) circle (1.5pt);
\draw [fill=qqqff] (0.,5.) circle (1.5pt);
\draw [fill=qqqff] (0.,4.) circle (1.5pt);
\draw [fill=qqqff] (1.,4.) circle (1.5pt);
\draw [fill=qqqff] (2.,4.) circle (1.5pt);
\draw [fill=qqqff] (2.,5.) circle (1.5pt);
\draw [fill=qqqff] (2.,6.) circle (1.5pt);
\draw [fill=qqqff] (1.,6.) circle (1.5pt);
\draw [fill=qqqff] (1.,5.) circle (1.5pt);
\draw [fill=qqqff] (4.,6.) circle (1.5pt);
\draw [fill=qqqff] (4.,5.) circle (1.5pt);
\draw [fill=qqqff] (4.,4.) circle (1.5pt);
\draw [fill=qqqff] (5.,4.) circle (1.5pt);
\draw [fill=qqqff] (5.,5.) circle (1.5pt);
\draw [fill=qqqff] (5.,6.) circle (1.5pt);
\draw [fill=qqqff] (6.,6.) circle (1.5pt);
\draw [fill=qqqff] (6.,5.) circle (1.5pt);
\draw [line width=1.2pt,color=qqqff] (6.,4.) circle (1.5pt);
\end{tikzpicture}
\end{center}
\caption{Equations equivalent to equation~\eqref{E:integrable_degree} \label{F:degeneration}}
\end{figure}
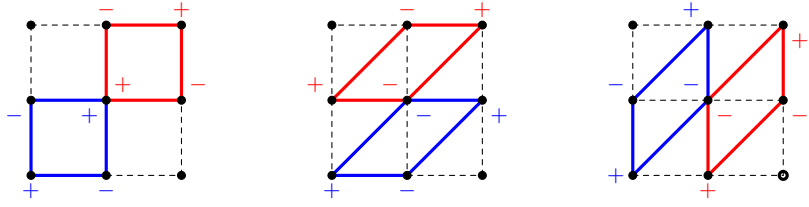
Therefore, we have Proposition~\ref{P:pro1}.
\begin{proposition}
\label{P:pro1}
Equation~\eqref{E:integrable_degree} is equivalent to any of the following equivalent statements
\begin{enumerate}[(i)]
\item $\bar{d}_{l+1,m+1}-\bar{d}_{l,m+1}-\bar{d}_{l+1,m}+\bar{d}_{l,m}=k(l-m)$,
\item $\bar{d}_{l+1,m+1}-\bar{d}_{l,m}-\bar{d}_{l,m+1}+\bar{d}_{l-1,m}=k(l)$,
\item $\bar{d}_{l+1,m+1}-\bar{d}_{l+1,m}-\bar{d}_{l,m}+\bar{d}_{l,m-1}=k(m)$,
\end{enumerate}
where $k$ is a function that only depends on $l-m$, $l$ ,$m$, respectively.
\end{proposition}
\begin{proof}
We prove that equation~\eqref{E:integrable_degree} is equivalent to $(i)$.
We denote the RHS of~\eqref{E:relation1} as $k_{l,m}$. We have $k_{l-1,m-1}=k_{l,m}$. It implies that
$k_{l,m}=k_{l-m,0}$ if $l\geq m$ or $k_{l,m}=k_{0,m-l}$ if $l<m$. This means that $k_{l,m}$ depends only on $l-m$  i.e. $k_{l,m}=k(l-m)$.
On the other hand, if $k_{l,m}=k(l-m)$, we have $k_{l-1,m-1}=k_{l,m}$. This is exactly equation~\eqref{E:relation1} which is equivalent to equation~\eqref{E:integrable_degree}.

Similarly, one can prove that equation~\eqref{E:integrable_degree} is equivalent to $(ii)$  and~\eqref{E:integrable_degree} is equivalent to  $(iii)$.
Hence $(i),(ii),(iii)$ are all equivalent themselves.
\end{proof}

It has been proved  in \cite{RobertsTran} that  $\bar{d}_{l,m}$ has quadratic growth along the diagonals with unit slope when corner initial values are affine  in $w$. Therefore, in general equations $(i)$, $(ii)$ and  $(iii)$ give such quadratic growth.
However, there are special cases that give us linear growth.
We note that the $(1,-1)$ staircase version (or the $I_2$ initial boundary condition) of equations~\eqref{E:integrable_degree}, $(i)$, $(ii)-(iii)$ are
\begin{align}
\bar{d}_{n+4}-2\bar{d}_{n+3}+2\bar{d}_{n+1}-\bar{d}_n&=0,\label{E:reduce1}\\
\bar{d}_{n+2}-2\bar{d}_{n+1}+\bar{d}_{n}&=k_n,\label{E:reduce2}\\
\bar{d}_{n+3}-\bar{d}_{n+2}-\bar{d}_{n+1}+\bar{d}_n&=k_{n} \label{E:reduce3},
\end{align}
where $k_n=k_{n+2}$ for the second equation and $k_n=k_{n+1}$ for the third equation.

The first equation gives us the  solution $\bar{d}_n=c_1(-1)^n+c_2+c_3n+c_4n^2$
which is quadratic.
Thus, in order to get linear growth we need to have $c_4=0$. Suppose that initial values  for equation~\eqref{E:reduce1} are
 $\bar{d}_{i_0},\bar{d}_{i_0+1},\bar{d}_{i_0+2}$ and $\bar{d}_{i_0+4}$. Solving the system of linear equations for $c_i$, we get
 $c_4=0$ if and only if $\bar{d}_{i_0}+\bar{d}_{i_0+3}-\bar{d}_{i_0+1}-\bar{d}_{i_0+2}=0$, which is equation~\eqref{E:reduce3} when $k=0$.

For the last two equations if $k_n\neq 0$ then $\bar{d}_n$ has quadratic growth as $\bar{d}_{n+1}-\bar{d}_{n}$ is linear. On the other hand, if $k_n=0$, it is easy to see that $\bar{d}_n$ grows linearly as
the characteristic equations for the last two equations are $(\lambda-1)^2=0$ and $(\lambda-1)^2(\lambda+1)=0$.

Let $(i)H, (ii)H$ and $(iii)H$ be the associated homogeneous versions of equations $(i), (ii)$ and $(iii)$, respectively. 
For the corner boundary conditions $I_1$, we have the following Theorem.

\begin{theorem}
\label{T:homogenous}
\begin{enumerate}
Consider equations $(i)H, (ii)H, (iii)H$ holding for $l\geq l_0, m\geq m_0$.
\item
For  equation $(i)H$,  if initial values   $\{\bar{d}_{l_0,m}, \bar{d}_{l,m_0}, l\geq l_0, m\geq m_0\}$ are linear in $m,l$, respectively, then $\bar{d}_{l,m}$ grows linearly in both horizontal and vertical directions.
\item
For  equation $(iiH)$, if  initial values  $\{\bar{d}_{l_0-1,m},\bar{d}_{l_0,m}, \bar{d}_{l,m_0}, l\geq l_0, m\geq m_0\}$ are linear in $m,l$, respectively, then $\bar{d}_{l,m}$ grows linearly along the horizontal direction.
\item
For  equation $(iii)H$, if  initial values  $\{\bar{d}_{l_0,m}, ,\bar{d}_{l,m_0},\bar{d}_{l,m_0-1}, l\geq l_0, m\geq m_0\}$ are linear in $m,l$ respectively then $\bar{d}_{l,m}$ grows linearly along the vertical directionn.
\end{enumerate}
\end{theorem}
 \begin{proof}

1.  For the homogeneous equation $(i)$, we have $\bar{d}_{l,m}=\bar{d}_{l_0,m}+\bar{d}_{l,m_0}-\bar{d}_{l_0,m_0}$. Therefore, if $\bar{d}_{l_0,m}$ and $\bar{d}_{l,m_0}$
are linear functions in  $m$ and $l$ respectively, then $\bar{d}_{l,m}$ grows linearly.
\newline
2.  For $k\geq 1,j\geq 0$, we have
\[
\bar{d}_{l_0-1+j+k,m_0+j}-\bar{d}_{l_0-1+j+k-1,m_0+j}=\bar{d}_{l_0-1+k,m_0}-\bar{d}_{l_0-1+k-1,m_0}.
\]
Thus we have
\begin{multline*}
\bar{d}_{l_0-1+j+k,m_0+j}=\bar{d}_{l_0-1+k,m_0}-\bar{d}_{l_0-1,m_0}+(\bar{d}_{l_0,m_0+1}-\bar{d}_{l_0-1,m_0+1})+\ldots+\\
(\bar{d}_{l_0,m_0+j-1}-\bar{d}_{l_0-1,m_0+j-1})
+\bar{d}_{l_0,m_0+j}.
\end{multline*}
It shows that for each fixed $j$, $\bar{d}_{l,m_0+j}$ grows linearly along the horizontal direction.
\newline
3. For the third statement, we just need to swap $l$ and $m$ and then it becomes the second statement.
\end{proof}
Recall that the $(q,-p)$ reduction of a lattice equation, where $q,p$ are positive co-prime integers, gives us an ordinary difference equation of order $(p+q)$ for the variable $V_n:=u_{l,m}$, where $n=lp+mq+1$ cf. \cite{RobertsTran}. 
\begin{corollary}
For  $\gcd(q,p)=1, q,p>0$, the $(q,-p)$ reductions  of  the homogeneous equations $(i)H$, $(ii)H$ and $(iii)H$  provide linear growth of $\bar{d}_n$.
\end{corollary}

 In fact, one can easily  see  this result  via characteristic equations  of these reduced equations.  The  $(q,-p)$ reductions of  homogenous equations $(i)H$, $(ii)H$ and $(iii)U$
give us  the following ordinary linear difference equations and their characteristic equations
\begin{align}
\bar{d}_{n+p+q}-\bar{d}_{n+p}-\bar{d}_{n+q}+\bar{d}_n&=0    &\implies \quad\quad\ \lambda^{p+q}-\lambda^p-\lambda^q+1=0, \label{E:qp_reduce}\\
\bar{d}_{n+2p+q}-\bar{d}_{n+p+q}-\bar{d}_{n+p}+\bar{d}_n&=0 &\implies\quad \lambda^{2p+q}-\lambda^{p+q}-\lambda^p+1=0,\label{E:qp_reduce1}\\
\bar{d}_{n+2q+p}-\bar{d}_{n+p+q}-\bar{d}_{n+q}+\bar{d}_n&=0 &\implies\quad \lambda^{p+2q}-\lambda^{p+q}-\lambda^q+1=0.\label{E:qp_reduce2}
\end{align}
Each of  the characteristic equations has $1$ as a double root and other roots which are distinct  roots of unity.
\section{Linear growth of some equations}
We now provide some examples of linearizable equations previously identified in the literature  whose actual degrees satisfy the homogeneous equations $(i)H$, and a  reduction of  $(ii)H$  from Theorem \ref{T:homogenous}.
\subsection{Liouville equation \label{SS:LV}}
In this section, we study growth of degrees    discrete Liouville equation.
 The discrete Liouville equation is given as follows
\begin{equation}
u_{l,m}u_{l+1,m+1}(u_{l+1,m}+1)(u_{l,m+1}+1)-u_{l+1,m}u_{l,m+1}=0.
\end{equation}
The discrete Liouville equation was first introduced by Adler  and Startsev \cite{DiscreteLiouville}.
It is known that this equation is Darboux integrable and linearizable. Therefore, it should have linear growth. In fact, this equation is  equivalent to
Equation 22 in \cite{HydonViallet} which has been checked to have linear growth with  staircase initial values $I_2$ of Figure~\ref{F:initial_values}.

In terms of projective coordinates, we have
\begin{align}
x_{l+1,m+1}&=x_{l+1,m}x_{l,m+1}z_{l,m},\label{E:LV_x}\\
z_{l+1,m+1}&=x_{l,m}(x_{l+1,m}+z_{l+1,m})(x_{l,m+1}+z_{l,m+1}).\label{E:LV_z}
\end{align}
On the boundary  of the first quadrant, we have $\deg(x_{l,0})=\deg(z_{l,0})=\deg_(x_{0,m})=\deg(z_{0,m})=1$.

By direct calculation,  we find that $\gcd_{2,1}=x_{1,0}$, $\gcd_{1,2}=x_{0,1}$ and $\gcd_{2,2}=x_{1,0}^2x_{0,1}^2z_{0,0}=x_{1,1}x_{1,0}x_{0,1}$.
 More generally, this holds for any $2\times 1$,  $1\times 2$
and $2\times 2$ blocks. However, if we extend to  a $2\times 3$ or a $3\times 2 $ block, we have an additional factor  $x_{1,1}+z_{1,1}$ for the top right vertex, i.e.
$(x_{1,1}+z_{1,1})$ is a divisor of $\gcd_{2,3}$ and $\gcd_{3,2}$. This implies that for $l,m\geq 3$ we get $z_{l-1,m-1}|\gcd_{l,m}$.
We also find that
\begin{align*}
{\rm{gcd}}_{3,2}&=\frac{\gcd_{2,2}\gcd_{3,1}x_{2,1}(x_{1,1}+z_{1,1})}{\gcd_{1,1}},\\
{\rm{gcd}}_{2,3}&=\frac{\gcd_{2,2}\gcd_{1,3}x_{1,2}(x_{1,1}+z_{1,1})}{\gcd_{1,1}},\\
{\rm{gcd}}_{3,3}&=\frac{\gcd_{2,3}\gcd_{3,2}x_{2,2}z_{2,2}}{\gcd_{2,2}}.
\end{align*}
Thus, for  $l,m\geq -1$, we build the following recurrence
\begin{equation}
\label{E:LV_gcd}
            G_{l+1,m+1} = \left\{
             \begin{array}{ll}
         \gcd_{l+1,m+1} & \mbox{if}\  l< 1 \ \mbox{or}\ m< 1\ \mbox{or} \ l=1, m=1 \\
        \frac{G_{l+1,m}G_{l,m+1}x_{l,m}(x_{l,m-1}+z_{l,m-1})}{G_{l,m-1}}&\mbox{if}\ l=1,m>1,\\
        \frac{G_{l+1,m}G_{l,m+1}x_{l,m}(x_{l-1,m}+z_{l-1,m})}{G_{l-1,m}}&\mbox{if}\ l>1,m=1,\\
        \frac{G_{l+1,m}G_{l,m+1}x_{l,m}z_{l,m}}{G_{l,m}}&\mbox{if}\ l>1,m>1.
        \end{array}
         \right.
\end{equation}
Taking initial values as random polynomials of degree $1$ in $w$ with integer coefficients, we have checked for $l,m\leq 12$ that
$\gcd_{l,m}=G_{l,m}$.
Taking degrees of both sides of \eqref{E:LV_gcd}, for $l,m>1$ we get
\begin{align*}
g_{l+1,m+1}&=g_{l+1,m}+g_{l,m+1}+2d_{l,m}-g_{l,m}\\
&=d_{l+1,m}+d_{l,m+1}+d_{l,m}-\bar{d}_{l+1,m}-\bar{d}_{l,m+1}+\bar{d}_{l,m}\\
&=d_{l+1,m+1}-\bar{d}_{l+1,m}-\bar{d}_{l,m+1}+\bar{d}_{l,m},
\end{align*}
where we have used \eqref{E:total_degree}.
This leads to the homogeneous equation $(i)H$ of Theorem \ref{T:homogenous}.
\begin{equation}
\label{LV_bar_degree}
\bar{d}_{l+1,m+1}=\bar{d}_{l+1,m}+\bar{d}_{l,m+1}-\bar{d}_{l,m}.
\end{equation}
Using~\eqref{E:LV_gcd} and \eqref{E:total_degree} we have $\bar{d}_{1,m}=m+2=\bar{d}_{m,1}$ for $m\geq 1$ and $\bar{d}_{2,m}=\bar{d}_{m,2}=2(m+2)$. Using the recurrence relation~\eqref{LV_bar_degree}  we get $\bar{d}_{l,m}=2(l+m)$, for $l,m\geq 2$. It means  $\bar{d}_{l,m}$ grows linearly along the diagonal, horizontal and vertical directions.
\subsection{Ramani-Joshi-Grammaticos-Tamizhmani equation (RJGT)}
The second linearizable lattice equation that we consider here is the  non-autonomous equation given in \cite{RJGT}.
This equation  is the non-autonomous version of  a CAC equation found by Hietarinta \cite{Hietarinta2004} (after a homographic transformation).
The RJGT equation is given by
\begin{equation}
\label{E:RJGT}
(u_{l+1,m+1}+r_{l,m+1})u_{l+1,m}(s_{l,m}u_{l,m}+t_{l,m})=(u_{l+1,m}+r_{l,m})u_{l,m}(s_{l,m+1}u_{l,m+1}+t_{l,m+1}),
\end{equation}
where  parameters  $r,s,t$ are free functions in $l,m$.
Using the same method described in section~\ref{SS:LV}, one finds that $\gcd_{2,1}=x_{1,0}$ and $\gcd_{1,2}=s_{{0,1}}x_{{0,1}}+t_{{0,1}}z_{{0,1}}$.
Therefore, we can build the recurrence
\begin{equation}
\label{E:RJGT_gcd}
G_{l+1,m+1}=\frac{G_{l+1,m}G_{l,m+1}x_{l,m}(s_{l,m}x_{l,m}+t_{l,m}z_{l,m})}{G_{l,m}},
\end{equation}
 for $l,m\geq 2$.
 If $l=0$ or $m=0$ or $(l,m)=(1,1)$, we take $G_{l,m}=1$.
  If $l=1, m>1$ we take $G_{l,m}=G_{l,m-1}(s_{l-1,m-1}x_{l-1,m-1}+t_{l-1,m-1}z_{l-1,m-1})$.
  If $l>1, m=1$, we take $G_{l,m}=G_{l-1,m}x_{l-1,m-1}$.
  The degree relation associated with \eqref{E:RJGT_gcd} the same as equation~\eqref{LV_bar_degree}.
 We have checked that with random integers for $r, s, t$ at each edges and random initial values as polynomials of degree $1$ in $w$ that
    $G_{l,m}=\gcd_{l,m}$ (up to a constant factor) for $l,m\leq 12$.
Suppose this holds for all $l,m>1$, it is easy to see that $\bar{d}_{l,m}=l+m+1$ for $l,m\geq 1$. Hence, $\bar{d}_{l,m}$ grows linearly.
\begin{remark}
\label{R:QRT_constant_degrees}
The  following QRT-type equations  cf.\cite{LC_linear_2011,Sahadevan}
\begin{align*}
u_{12}&=\frac{u_1+u_2-(1-u_1u_2)u}{1-u_1u_2+(u_1+u_2)u},\\
u_{12}&=\frac{u_1-u_2+(1+u-1u_2)u}{1+u_1u_2-(u_1-u_2)u}
\end{align*}
also behave similarly as in \eqref{E:RJGT_gcd} by replacing $x_{l,m}(s_{l,m}x_{l,m}+t_{l,m}z_{l,m})$ with $x^2_{l,m}+z^2_{l,m}$.
In fact, except for the boundary condition, these equations have constant degree $\bar{d}_{l,m}=3$ for $l,m\geq 1$.
By using the transformation $v=\arctan(u)$, these equations can be brought to the following  linear equations
\begin{align*}
v+v_1-v_2-v_{12}=p\pi,\\
v-v_1-v_2+v_{12}=p\pi,
\end{align*}
where $p\in\Z$.
\end{remark}
\begin{remark}
In \cite{solvable_chaos} the authors gave the following example which is chaotic and linearizable (under the same transformation above)
\begin{equation}
\label{E:QRT_type_reduction}
x_{n+1}=\frac{3x_n-x_n^3-x_{n-1}(1-3x_n^2)}{1-3x_n^2+(2x_n-x_n^3)x_{n-1}}.
\end{equation}
In fact it can be seen as the  $(1,-1)$ reduction $u_1=u_2$ of the associated lattice equation derived from  the linear equation $ v+av_1+bv_2+cv_{12}=p\pi$,
where $p\in\Z$ and $a=-1,b=-2,c=1$ or $a=-2,b=-1,c=1$ via the transformation $u=\tan(v)$.
We also note that the other two QRT type equations in \cite{LC_linear_2011} derived from the cases where $a=-3,b=3,c=1$ and  $a=-2, b=1,c=1$ by the same
 transformation given in  Remark \ref{R:QRT_constant_degrees}  have exponential growth. This shows that  algebraic entropy is not preserved under non-rational transformation.
\end{remark}

\subsection{The $(3,1)$-reduction of $H_1$}
The equation obtained from the $(3,1)$-reduction is given by
\begin{equation}
\label{E:3_1_reduction}
(u_{n+4}-u_n)(u_{n+3}-u_{n+1})=\alpha.
\end{equation}
This gives us $u_{n+4}=u_n+\alpha/(u_{n+3}-u_{n+1})$.
This reduction is an exceptional case which was shown to be linearizable  \cite{Kamp_growth}.

We can introduce projective coordinates $u_n=x_n/z_n$, and obtain the followings rules
\begin{align}
x_{n+4}&=-\alpha\,z_{{n+3}}z_{{n+1}}z_{{n}}+x_{{n}}x_{{n+1}}z_{{n+3}}-x_{{n}}x_
{{n+3}}z_{{n+1}},\label{E:3_1_x}\\
z_{n+4}&=z_{{n}} \left( x_{{n+1}}z_{{n+3}}-x_{{n+3}}z_{{n+1}} \right)
\label{E:3_1_z}.
\end{align}
We start with $(x_i,z_i)$ for $1\leq i\leq 4$ as initial values. By using these rules, we can calculate $(x_n,z_n)$ for $n\geq 5$  as functions of  initial values.
We can easily see that they are polynomials in $x_1,x_2,x_3,x_4,z_1,z_2,z_3,z_4$.
We assume that ${\rm{deg}}(x_i)={\rm{deg}}(z_i)=1$ for $1\leq i\leq 4$.
Denote $d_n={\rm{deg}}(x_n)={\rm{deg}}(z_n)$.

 We have $d_{n+4}=d_n+d_{n+1}+d_{n+3}$.
By direct calculation, we find that $\gcd(x_8,z_8)=x_2z_4-x_4z_2$ and $(x_2z_4-x_4x_2)^2(x_5z_3-z_5x_3)| \gcd(x_9,z_9)$.
Therefore, we denote $\gcd_n=\gcd(x_n,z_n)$ and we write $x_n=\gcd_n \bar{x}_n$ and $z_n=\gcd_n \bar{z}_n$.
Let $\bar{d}_n=\deg\bar{x}_n=\deg\bar{z}_n$ and $g_n=\deg\gcd_n$.
We know that $\gcd_n \gcd_{n+1}\gcd_{n+3}|\gcd_{n+4}$. Denote
$\overline{\gcd}_{n+4}=\gcd_{n+4}/(\gcd_n \gcd_{n+1}\gcd_{n+3})$ and
$A_n=x_{n-5}z_{n-7}-z_{n-5}x_{n-7}$.
For $n\geq 5$, we have
$A_{n+4}^2 A_{n+5}|\gcd_{n+4}$.
We can see that $\gcd(A_{n+4}^2A_{n+5},\gcd_n \gcd_{n+1}\gcd_{n+3})=g_{n-1}g_{n-3}g_{n}g_{n-2}  A_{n+4}B_{n+4}$,
Therefore for $n\geq 5$, we have
\begin{align*}
\frac{A_{n+4}^2A_{n+5}\gcd_n \gcd_{n+1}\gcd_{n+3}}{\gcd(A_{n+4}^2A_{n+5},\gcd_n \gcd_{n+1}\gcd_{n+3}}
&=\frac{A_{n+4}A_{n+5}\gcd_{n+3}\gcd_{n}\overline{\gcd}_{n+1}}{\gcd_{n-1}B_{n+4}}.
\end{align*}
It suggests that we should try the following recursive formula
\begin{equation}
\label{E:3_1_relation}
G_{n+4}=\frac{A_{n+4}A_{n+5}G_{n+3}G_{n}}{G_{n-1}}=\frac{(x_{n-1}z_{n-3}-z_{n-1}x_{n-3})(x_nz_{n-2}-z_nx_{n-2})G_{n+3}G_{n}}{G_{n-1}},
\end{equation}
for $n>4$ and $G_n=\gcd_n$ for $n\leq 8$.

We now take initial values as random polynomials of degree $1$ in $w$. We have checked for $n\leq 40$ that $G_n=\gcd_{n}$ (up to a constant factor).
Thus, we conjecture that $\gcd_n=G_n$. Taking degrees of both sides of \eqref{E:3_1_relation}, we obtain
\begin{align*}
g_{n+4}&=d_{n-1}+d_{n-3}+d_n+d_{n-2}+g_{n+3}+g_{n}-g{n-1}\\
&=\bar{d}_{n-1}+d_{n+1}+d_{n+3}+d_n-\bar{d}_n-\bar{d}_{n+3}\\
&=\bar{d}_{n-1}+d_{n+4}-\bar{d}_n-\bar{d}_{n+3}.
\end{align*}
This implies that
\begin{equation}
\label{E:3_1_bar_relation}
\bar{d}_{n+4}=\bar{d}_n+\bar{d}_{n+3}-\bar{d}_{n-1},
\end{equation}
which is equation~\eqref{E:qp_reduce1} with $(q,p)=(3,1)$, i.e. a reduction of $(ii)H$.
The characteristic equation for this linear equation is
\begin{equation}
\lambda^5-\lambda^4-\lambda+1=(\lambda-1)^2(\lambda-1)(\lambda^2+1)=0.
\end{equation}
This equation has the following roots: $1$ (double root), $-1,i, -i$. This means $\bar{d}_n$ grows linearly for $n\geq 5$.
On the other hand, we know that $\bar{d}_i=1$ for $1\leq i\leq 4$, and $\bar{d}_i=2i-7$ for $i=5\leq i\leq 8$. It is easy to prove that $\bar{d}_n=2n-7$
for $n\geq 5$. It again confirms that the  sequence $\bar{d}_n$  has linear growth.

\section{Searching for lattice equations with linear growth}
In this section, we use the homogeneous equations  given in Theorem~\ref{T:homogenous} to search for  examples of lattice equations with linear growth.
We then confirm they are linearizable.
We start with a general form of a multi-affine equation on quad-graphs
\begin{multline}
Q:a_{{15}}uu_{{1}}u_{{2}}u_{{12}}+a_{{11}}uu_{{1}}u_{{2}}+a_{{12}}uu_{{1
}}u_{{12}}+a_{{13}}uu_{{2}}u_{{12}}+a_{{14}}u_{{1}}u_{{2}}u_{{12}}+a_{
{5}}uu_{{1}}+a_{{6}}uu_{{2}}+a_{{7}}uu_{{12}}+\\
a_{{8}}u_{{1}}u_{{2}}
+a_
{{9}}u_{{1}}u_{{12}}+a_{{10}}u_{{2}}u_{{12}}+a_{{1}}u+a_{{2}}u_{{1}}+a
_{{3}}u_{{2}}+a_{{4}}u_{{12}}+a_{{0}}=0.\qquad
\end{multline}
As the simplest case, we can search for  autonomous equations that satisfy the relation~$(i)H$.
Going backwards from this equation, we obtain
\[
g_{l+1,m+1}=g_{l+1,m}+g_{l,m+1}+2d_{l,m}-g_{l,m}.
\]
This suggests that
\begin{equation}
\label{E:backward1}
{\rm{gcd}}_{l+1,m+1}=\frac{\gcd_{l+1,m}\gcd_{l,m+1}A_{l+1,m+1}}{\gcd_{l,m}},
\end{equation}
where $\deg(A_{l+1,m+1})=2d_{l,m}$. The obvious choice is
$A_{l+1,m+1}=f^{(1)}x_{l,m}^2+f^{(2)} x_{l,m}z_{l,m}+f^{(3)} z_{l,m}^2=(t_1 x_{l,m}+t_2z_{l,m})(s_1x_{l,m}+s_2z_{l,m}).$
We also assume that the $\gcd$ first appears at $(2,1)$ and $(1,2)$ and $t_1=s_1=1$ if $t_1,s_1\neq 0$.
The search algorithm can be broken down to the following steps.
\begin{enumerate}
\item Write the rule in projective coordinates and calculate  $x$ and $z$ at vertices $(2,1)$ and $(1,2)$ as polynomials in initial values.
\item At the point $(2,1)$ and $(1,2)$, substitute $x_{1,0}=-t_2z_{1,0}/t_1$ and  $x_{0,1}=-s_2z_{1,0}/s_1$ into $x$ and $z$ and collect all the coefficients.
\item Set these coefficients to $0$ and solve for $a_0,a_1,\ldots, a_{15}$.
\item Substitute solutions back to $Q$ and eliminate equations that are degenerate.
\item For the `survival equations', check the recurrence relation~\eqref{E:backward1}.
\end{enumerate}
 For example, we obtain the following equation which grows linearly
 \begin{multline*}
 Q_7:\ u_{{1}}u_{{2}}u{a_{{11}}}^{2}+u_{{1}}uu_{{12}}a_{{11}}a_{{12}}+
 \left( {a_{{11}}}^{2}s_{{2}}+a_{{6}}a_{{12}} \right) u_{{1}}u+u_{{2}}
ua_{{6}}a_{{11}}+ua_{{11}}a_{{12}}t_{{2}}u_{{12}}+
\\
 \left( a_{{6}}a_{{
11}}s_{{2}}+a_{{6}}a_{{12}}t_{{2}} \right) u+
{a_{{11}}}^{2}t_{{2}}u_{{
1}}u_{{2}}+a_{{11}}a_{{12}}s_{{2}}u_{{1}}u_{{12}}+ \left( {a_{{11}}}^{
2}s_{{2}}t_{{2}}+a_{{6}}a_{{12}}s_{{2}} \right) u_{{1}}+
\\
a_{{6}}a_{{11}
}t_{{2}}u_{{2}}+a_{{11}}a_{{12}}s_{{2}}t_{{2}}u_{{12}}+a_{{6}}a_{{11}}
s_{{2}}t_{{2}}+a_{{6}}a_{{12}}s_{{2}}t_{{2}}=0,\qquad\qquad
 \end{multline*}
 where $u=u_{l,m}, u_1=u_{l+1,m}, u_2=u_{l,m+1}, u_{12}=u_{l+1,m+1}$.
 This equation can be written as
 \begin{equation}
 \label{E:Q7}
 a_{11}(u+t_2)(u_2+s_2)(a_{11}u_{1}+a_6)=-a_{12}(u_1+t_2)(u+s_2)(a_{11}u_{12}+a_6).
 \end{equation}
 If $a_{11}=-a_{12}$, this equation is an autonomous version of equation (13) in \cite{RJGT} which is linearizable.
 We write equation~\eqref{E:Q7} as follows
 \begin{equation}
 \frac{(a_{11}u_{12}+a_6)}{(a_{11}u_{1}+a_6)}\frac{(u+s_2)}{(u_2+s_2)}=-\frac{a_{11}(u+t_2)}{a_{12}(u_1+t_2)}.
 \end{equation}
    This equation can reduce to the discrete Burgers equation \cite{LC_linear} by first taking $a_{12}=-a_{11}, a_6=pa_{11}, t_2=0, s_2=a_{11}$ and
 taking the limit as $a_{11}\rightarrow \infty$.

 By introducing $u=x_2/x-t_2$, we obtain the following equation
 \[
 \frac{S_2(f(x,x_1,x_2,x_12))}{f(x,x_1,x_2,x_{12}})=\frac{-a_{11}}{a_{12}},
 \]
 where $S_2$ denotes the shift in the second direction and
 \begin{equation*}
 f(x,x_1,x_2,x_{12})=\frac{(a_{11}u_1+a_6)x_1}{(u+s_2)x}
 =\frac{a_{11}x_{12}-t_2x_1+a_6x_1}{x_2-t_2x+s_2x}.
 \end{equation*}
 It implies that
 $$
 f(x,x_1,x_2,x_{12}=c_0\left(\frac{-a_{11}}{a_{12}}\right)^m.
 $$
 This equation can be written in a linear form as follows
 \[
 a_{11}x_{12}-t_2x_1+a_6x_1=c_0\left(\frac{-a_{11}}{a_{12}}\right)^m\left(x_2-t_2x+s_2x\right).
 \]

 We note that this linearization process also holds for an autonomous version of~\eqref{E:Q7} if $s_2(l,m)$, ${a_6}(l,m)/{a_{11}}(l,m)$ depend on $m$ only and $a_{11}/a_{12}=-k$  and $t_2$ are constant. This non-autonomous equation has the form:
  \begin{equation}
 \label{E:non_Q7}
 \frac{({a_{11}}(l,m+1)u_{12}+{a_6}(l,m+1))}{({a_{11}}(l,m)u_{1}+{a_6}_(l,m))}\frac{(u+{s_2}(l,m))}{(u_2+{s_2}(l,m+1))}=\frac{k(u+t_2)}{(u_1+t_2)}.
 \end{equation}
 This equation is actually equation~\eqref{E:RJGT} if $k=1$.

 We have also found that equation~\eqref{E:non_Q7} does not satisfy condition 60a, 60b, 60c in \cite{LC_linear}, so it
 cannot be linearized by an invertible point transformation. However, as we have seen above, it is linearized by the
 Cole-Hopf transformation. This equation can be linearized by writing it as a system of two equations.

 For the case $k=1$, equation~\eqref{E:non_Q7} can be written as
\begin{equation}
\label{E:Q7special}
u(u_2+s_2)(u_1+t)=u_1(u+s)(u_{12}+t_2),
\end{equation}
where $s_2=s_{l,m+1},\  t_2=t_{l,m+1}$. We follow the method given in \cite{LC_linear} to linearize this equation.
We introduce a potential function $v$ that satisfies the following
\begin{equation}
\label{E:potential}
v_{l,m+1}=u_{l,m}v_{l,m}=\mathcal{E}^{(2)}_{l,m}v_{l,m},\
v_{l+1,m}=\frac{(u_{l,m}+s_{l,m})}{(u_{l+1,m}+t_{l,m})}v_{l,m}=\mathcal{E}^{(1)}_{l,m}v_{l,m}.
\end{equation}
The compatibility of these two equations gives us equation~\eqref{E:Q7}
\begin{equation}
\frac{\mathcal{E}^{(1)}_{l,m+1}}{\mathcal{E}^{(1)}_{l,m}}-\frac{\mathcal{E}^{(2)}_{l+1,m}}{\mathcal{E}^{(2)}_{l,m}}=0.
\end{equation}
We look for a symmetry generator that has the form
\begin{equation}
\hat{X}_{l,m}=\phi_{l,m}(u_{l,m},v_{l,m})\partial_{u_{l,m}}+\psi_{l,m}(u_{l,m},v_{l,m})\partial_{v_{l,m}}.
\end{equation}
For simplicity, we write $(l,m)$ as $(0,0)$. Thus, we obtain the following equation
\begin{align}
\psi_{0,1}(u_{0,1},v_{0,1})&=v_{0,0}\phi_{0,0}(u_{0,0},v_{0,0})+\psi_{0,0}(u_{0,0},v_{0,0})u_{0,0},\label{E:condition1}\\
\psi_{1,0}&=\frac{\phi_{0,0}v_{0,0}}{u_{1,0}+t_{0,0}}-\frac{\phi_{1,0}v_{0,0}(u_{0,0}+s_{0,0})}{(u_{1,0}+t_{0,0})^2}+\frac{\psi_{0,0}(u_{0,0}+s_{0,0})}{u_{1,0}+t_{0,0}}.\label{E:condition2}
\end{align}
Since $v_{0,1}=u_{0,0} v_{0,0}$, the first equation implies that $\psi$ only depends on the second variable, i.e. $\psi_{0,0}(u_{0,0},v_{0,0})=\psi_{0,0}(v_{0,0})$.
Therefore, we have
\begin{equation}
\phi_{0,0}=\frac{\psi(v_{0,1})-u_{0,0}\psi_{0,0}(v_{0,0})}{v_{0,0}}.
\end{equation}
Using this formula, we substitute $\phi_{0,0}$, $\phi_{1,0}$ and $v_{1,0}=\mathcal{E}^{(1)}_{0,0}v_{0,0}$ into equation~\eqref{E:condition2} to obtain
\begin{equation}
\label{E:linear}
s_{0,0}\psi_{0,0}-t_{0,0}\psi_{1,0}+\psi_{0,1}-\psi_{1,1}=0.
\end{equation}
Therefore, we take $\psi_{0,0}=w_{0,0}$, and
$\phi_{0,0}=(w_{0,1}-u_{0,0}w_{0,0})/v_{0,0}$ where $w_{l,m}$ satisfies the linear equation~\eqref{E:linear}.
Thus, it suggests we take ${\bf w}=(w_{0,0},w_{0,1})$.
Using Theorem 6 in \cite{LC_linear}, we have $\gamma^{(1)}=1, \gamma^{(2)}=0,\beta^{(1)}=-u_{0,0}/v_{0,0},\beta^{(2)}=1/v_{0,0}$.
Substituting in equation~(38) in \cite{LC_linear}, we obtain the transformation
\[
z^{(1)}_{0,0}=K^{(1)}_{0,0}(u_{0,0},v_{0,0}),\ z^{(2)}_{0,0}=K^{(2)}_{0,0}(u_{0,0},v_{0,0}),
\]
where $K^{(1)}=v_{0,0}$ and $K^{(2)}=v_{0,0}u_{0,0}$. This gives us the Cole-Hofp transformation
$u_{l,m}=z^{(1)}_{l,m+1}/z^{(1)}_{l,m}$ which linearizes the original quad-graph equation.
It is easy to check that by using the transformation $z^{(1)}_{l,m}=v_{l,m}$ and $z^{(2)}_{l,m}=v_{l,m}u_{l,m}$, we bring a system of equations~\eqref{E:potential}
to a linear equation
\begin{equation}
L_{l,m}: s_{l,m}w-t_{l,m}w_1+w_2-w_{12}=0,
\end{equation}
where $z^{(1)}_{l,m}$ satisfies $L_{l,m}=0$ and $z^{(2)}_{l,m}$ satisfies $L_{l,m+1}=0$, i.e. we shift parameters in the  second direction by $1$.

 For the relation~$(ii)H$, we get
 \begin{align*}
 g_{l+1,m+1}&=g_{l,m}+g_{l,m+1}+g_{l,m+1}+\bar{d}_{l-1,m}+\bar{d}_{l+1,m}\\
 &=g_{l,m}+g_{l,m+1}+d_{l+1,m}+d_{l-1,m}-g_{l-1,m}.
 \end{align*}
 We can use the  similar argument for relation~$(iii)H$. Thus, we  can take respectively for $(ii)H$ and $(iii)H$:
 \begin{align*}
 {\rm{gcd}}_{l+1,m+1}&=\frac{\gcd_{l,m+1}gcd_{l,m}A_{l+1,m+1}}{\gcd_{l-1,m}},\\
 {\rm{gcd}}_{l+1,m+1}&=\frac{\gcd_{l+1,m}gcd_{l,m}A_{l+1,m+1}}{\gcd_{l,m-1}},
 \end{align*}
 where
 \begin{align*}
 \deg(A_{l+1,m+1})&=d_{l+1,m}+d_{l-1,m}=(d_{l-1,m}+d_{l,m-1})+(d_{l,m}+d_{l+1,m-1}),\\
 \deg(A_{l+1,m+1})&=d_{l,m+1}+d_{l,m-1}=(d_{l,m-1}+d_{l-1,m})+(d_{l,m}+d_{l-1,m+1}).
 \end{align*}
 For the former case, one can try $A_{l+1,m+1}=B_{l,m+1}B_{l+1,m+1}$, where $\deg(B_{l+1,m+1})=(d_{l,m}+d_{l+1,m-1})$.
 For the latter case, one can try $A_{l+1,m+1}=B_{l+1,m}B_{l+1,m+1}$, where $\deg(B_{l+1,m+1})=(d_{l,m}+d_{l-1,m+1})$.
 The search for equations that behave similarly to the $(3,1)$  reduction of $H_1$ or the former case does not give non-degenerate equations.

 \section{Conclusion}
 In this paper, based on the recurrence relation~\eqref{E:integrable_degree} for the degrees of many integrable lattice equations~\cite{RobertsTran}, we derived some linear recurrences $(i)H, (ii)H$ and $(ii)H$ of Theorem~\ref{T:homogenous} that imply linear degree growth.
 We then used these recurrences to build recursive formulas for the gcds. Thus, we were able to search for examples of linearizable equations with certain gcd patterns, for example
 the RJGT equation and discrete Burgers equation. A symmetry method  given in \cite{LC_linear} was used to linearize the RJGT equation.
 Moreover, we have also noted that there are linearizable equations with exponential growth \cite{solvable_chaos, Sahadevan}.  This  is because the transformations to bring these equations to linear equations are not rational.

 We have also noticed that some other linearizable equations such as  Equation 15 in \cite{HV} where $p_6=0$, another form of  Liouville equation \cite{LC_linear_2011} and Equation 39   with conditions (40) in \cite{LC_linear} satisfy  the recurrence~\eqref{E:backward1} and hence satisfy $(i)H$. Moreover, we have not found a lattice equation whose gcd relations give us  the homogeneous equations $(ii)H$ or $(iii)H$ of Theorem \ref{T:homogenous} directly. Thus, it leads to a question whether all lattice equations which are linearized via a rational transformation obey the recurrence~$(i)H$.
It is also worth studying why those equations  would share one and the same linear recurrence.

\section*{Acknowledgements}
This research is supported by the  Australian Research Council.

\end{document}